\documentclass[5p,preprint,authoryear,twocolumn,times]{elsarticle}

\usepackage[T1]{fontenc}
\usepackage{newtxtext}
\usepackage[varvw,vvarbb]{newtxmath}
\usepackage{mathtools}
\usepackage{graphicx,color}
\usepackage[tight,hang]{subfigure}
\usepackage{arydshln}
\usepackage{url}

\newtheorem{theorem}{Theorem}[section]
\newtheorem{lemma}[theorem]{Lemma}
\newtheorem{proposition}[theorem]{Proposition}
\newtheorem{corollary}[theorem]{Corollary}
\newdefinition{example}{Example}
\newdefinition{remark}{Remark}
\newproof{proof}{Proof}

\DeclareMathOperator\rank{rank}
\DeclareMathOperator\nrank{n\hspace{.5pt}rank}
\DeclareMathOperator\im{Im}
\DeclareMathOperator\real{Re}
\DeclareMathOperator\spec{spec}
\DeclareMathOperator\diag{diag}
\newcommand\abs[1]{\ensuremath{\lvert#1\rvert}}
\newcommand\sigmin[1]{\ensuremath{\underline{\sigma\mkern-3mu}\mkern3mu}\!\left(#1\right)}
\newcommand\mmatrix[2][cccccc]{\ensuremath{\left[\begin{array}{#1}#2\end{array}\right]}}
\newcommand\smat[1]{\ensuremath{\bigl[\begin{smallmatrix}#1\end{smallmatrix}\bigr]}}%
\newcommand\pdiri[2]{\ensuremath{\mathrm{pdir_i}\!\left(#1,#2\right)}}
\newcommand\pdiro[2]{\ensuremath{\mathrm{pdir_o}\!\left(#1,#2\right)}}
\newcommand\zdiri[2]{\ensuremath{\mathrm{zdir_i}\!\left(#1,#2\right)}}
\newcommand\zdiro[2]{\ensuremath{\mathrm{zdir_o}\!\left(#1,#2\right)}}
\newcounter{cAss}
\newcounter{cAssSaved}
\newcommand\Ass[1]{\ensuremath{\boldsymbol{\mathcal A}_{\text{\hspace{0.75pt}\bf#1}}}}
\newlength\asswidth
\newenvironment{assumptions}%
 {\begin{list}{\Ass{\arabic{cAss}}:}{%
 \setcounter{cAssSaved}{\thecAss}\usecounter{cAss}\setcounter{cAss}{\thecAssSaved}
 \setlength\topsep{.667ex plus 2pt minus 2pt}
 \setlength\itemsep\topsep
 \settowidth\asswidth{\Ass{\arabic{cAss}}:}
 \addtolength\asswidth\labelsep
 \setlength\leftmargin\asswidth
 \setlength\labelwidth\asswidth}}%
 {\end{list}}
\renewcommand\cite{\citep}

\begin{document}

\begin{frontmatter}

\title{On the Internal Stability of Diffusively Coupled Multi-Agent Systems\\ and the Dangers of Cancel Culture\tnoteref{footnoteinfo}} 

\tnotetext[footnoteinfo]{Supported by the Israel Science Foundation (grants 3177\hspace{.5pt}/21 and 2285\hspace{.5pt}/20) and Sakranut Graydah at the Technion.}

\author[ME]{Gal Barkai}    \ead{galbarkai@campus.technion.ac.il}
\author[ME]{Leonid Mirkin} \ead{mirkin@technion.ac.il}
\author[AE]{Daniel Zelazo} \ead{dzelazo@technion.ac.il}

\address[ME]{Faculty of Mechanical Engineering, Technion---IIT, Haifa 3200003, Israel}
\address[AE]{Faculty of Aerospace Engineering, Technion---IIT, Haifa 3200003, Israel}

\begin{abstract}
 We study internal stability in the context of diffusively-coupled control architectures, common in multi-agent systems (i.e.\ the celebrated consensus protocol), for linear time-invariant agents. We derive a condition under which the system can not be stabilized by any controller from that class. In the finite-dimensional case the condition states that diffusive controllers cannot stabilize agents that share common unstable dynamics, directions included. This class always contains the group of homogeneous unstable agents, like integrators. We argue that the underlying reason is intrinsic cancellations of unstable agent dynamics by such controllers, even static ones, where directional properties play a key role. The intrinsic lack of internal stability explains the notorious behavior of some distributed control protocols when affected by measurement noise or exogenous disturbances.
\end{abstract}

\begin{keyword}
 Multi-agent systems, controller constraints and structure, stability.
\end{keyword}

\end{frontmatter}

%%%%%%%%%%%%%%%%%%%%%%%%%%%%%%%%%%%%%%%%%%%%%%%%%%%%%%%%%%%%%%%%%%%%%%%%%%%%%%%%%%%%%

\section{Introduction} \label{sec:intro}

A multi-agent system (MAS) is a collection of independent systems (agents) coupled via the pursuit of a common goal. In large-scale MASs the information exchange between agents might be costly. As such, it is commonly limited to a subset of agents, known as \emph{neighbors}. Control laws that use only information from neighboring agents are called \emph{distributed}.

This work studies a class of distributed control laws, where only \emph{relative} measurements are exchanged between neighbors. In other words, each agent has access only to the difference between its output and that of each of its neighbours. Such control laws are called \emph{diffusive} and systems controlled by them are known as \emph{diffusively coupled}. Diffusive control laws are common in the MAS literature.  Relative sensing appears naturally in MAS tasks, where absolute measurements are hard to obtain, such as space and aerial exploration and sensor localization, see \cite{SH:05,KKM:09,ZM:11rsn} and the references therein. The consensus and synchronization problems are well-known examples of diffusively coupled systems \cite{O-SFM:07,WSA:11}.

However, diffusively-coupled systems behave poorly when affected by disturbances and noise. Measurement noise rapidly deteriorates performance \cite[\S III-A]{ZM:11} and even dynamic controllers can hardly attenuate disturbances \cite{Ding:15}. To cope with the difficulties, different relaxing assumptions are assumed. Some allow for non-relative state \cite{YE:12} or output \cite{MG:19} measurements, while others employ an undisturbed leader \cite{Ding:15} or impose limitations even on bounded disturbances \cite[Prop.\,5]{BDeP:15}. Despite these different assumptions, if they fail, the resulting trajectories exhibit certain common traits that can be associated with instability. These traits can be illustrated by the classical consensus protocol, considered below for a set of integrator agents and with a static interaction network.

\subsection{Motivating example} \label{sec:intro:mot}

Reaching agreement between autonomous agents is a fundamental building block in multi-agent coordination \cite{RB:08}. In its simplest form it studies a group of independent integrator agents $\dot x_i(t)=u_i(t)$, where $x_i$ and $u_i$ are their states and control inputs, respectively. The goal is to reach asymptotic agreement between all agents, in the sense that
\begin{equation} \label{eq:agreement}
 \lim_{t\to\infty}\bigl(x_i(t)-x_j(t)\bigr)=0,\quad\forall i,j,
\end{equation}
under the constraint that the $i$th agent has access only to states of its neighbors, whose indices belong to a set $\mathcal N_i$. This problem can be solved by the celebrated consensus protocol \cite{O-SFM:07}, which is a diffusive state-feedback of the form
\begin{equation}\label{eq:conprot}
 u_i(t)=-\sum_{\mathclap{j\in\mathcal N_i}}\bigl(x_i(t)-x_j(t)\bigr),\quad\forall i.
\end{equation}
If certain connectivity conditions on the communication topology hold (i.e.\ connectedness), then the control law \eqref{eq:conprot} drives the agents to agreement exponentially fast \cite[Ch.\,3]{ME:10}. The state trajectories of four agents controlled by \eqref{eq:conprot} are shown in Fig.\,\ref{fig:stepD} in the time interval $[0,t_d]$. Observe that on this time interval the states converge exponentially to the average of their initial conditions and the control signals all asymptotically vanish.

\begin{figure}[!t]
 \centering
  \subfigure[The states response.]{\label{fig:stepDx}\includegraphics[width=0.48\columnwidth,clip]{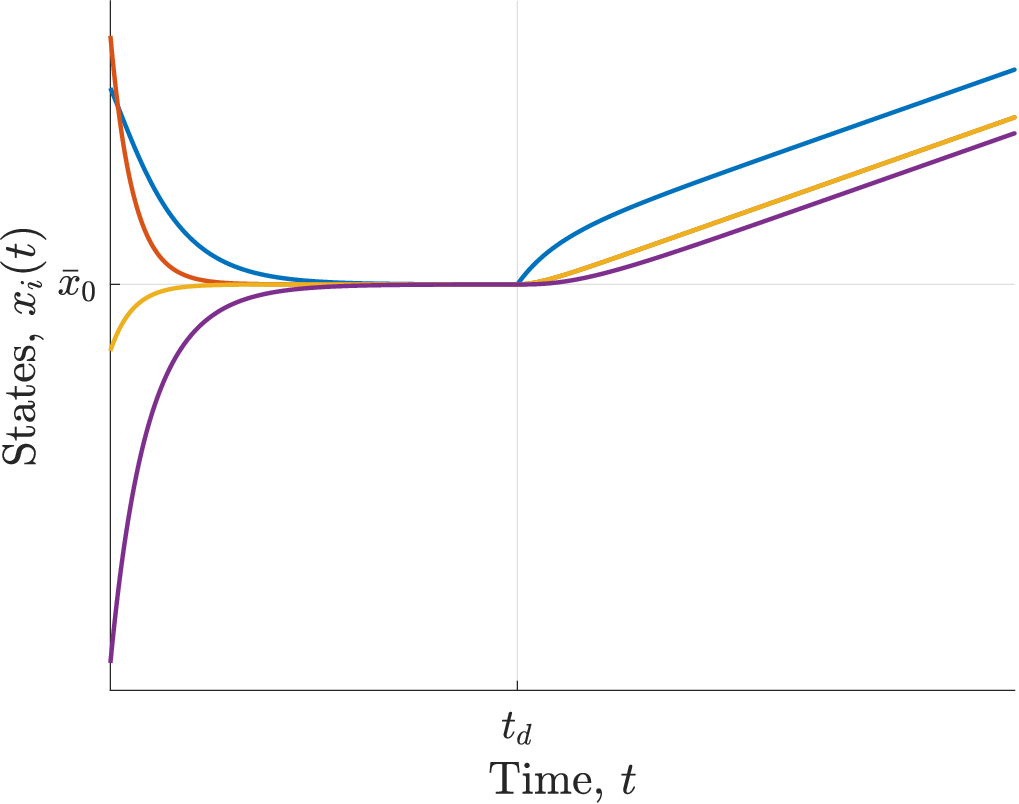}}
  \hspace{\stretch1}
  \subfigure[The control signals response.]{\label{fig:stepDu}\includegraphics[width=0.48\columnwidth,clip]{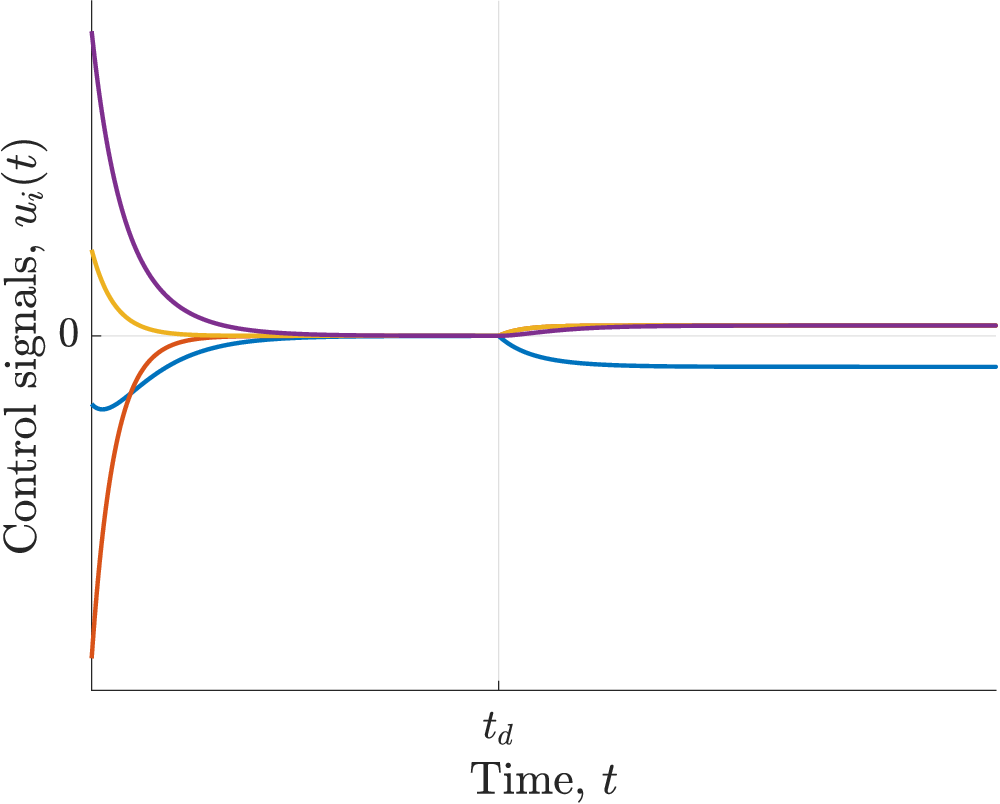}
  }
 \caption{Simulation of protocol \eqref{eq:conprot} perturbed by a step at $t=t_d$.} \label{fig:stepD}
\end{figure}

This might no longer be the case if the agents are affected by load disturbances $d_i$, viz.
\begin{equation} \label{eq:distcons}
 \dot x_i(t)=u_i(t)+d_i(t).
\end{equation}
An example of what happens in such situations is also shown in Fig.\,\ref{fig:stepD}.  At the time instance $t=t_d$ one agent is affected by a unit step disturbance. As a result, all states cease to agree and start to diverge when $t>t_d$, whereas the control signals reach non-zero steady-state values.

The apparent instability of the whole system, manifested in the unboundedness of the states, can be explained by the well-known fact that the consensus protocol has a closed-loop eigenvalue at the origin \cite{O-SFM:07}. Nevertheless, the boundedness of the control signals under such conditions is intriguing. Situations wherein some signals in the closed-loop system are bounded while some others are not normally indicate unstable \emph{pole-zero cancellations} in the feedback loop \cite[Sec.\,5.3]{ZDG:95}. However, controller \eqref{eq:conprot} is static and thus has no zeros.

\subsection{Contribution}

The example above suggests that a deeper inspection of the \emph{internal stability} property could offer insight into the behavior of diffusively-coupled systems. The internal stability of any feedback interconnection requires the stability of all possible input\,/\,output relations in the system, see \cite{ZDG:95,SkP:05}. However, to the best of our knowledge, internal stability has not been explicitly studied in the context of diffusively-coupled architectures of MASs yet.

In this paper we show that diffusively-coupled systems of LTI (linear time-invariant) agents might not be internally stabilizable. Loosely speaking, this happens if the agents share common unstable dynamics, directions counting. This, for example, is always the case in a group of homogeneous unstable agents, like those discussed in \S\ref{sec:intro:mot}.

When restricting the result to finite-dimensional agents, we also explain the mechanism behind the shown internal instability. It is caused by \emph{unstable cancellations} in the cascade of the aggregate plant and a diffusive controller. Important is that these cancellations are caused not by controller zeros, but rather by an intrinsic spatial deficiency of the diffusively-coupled configuration. These cancellations are intrinsic to the diffusive structure and cannot be affected by controller dynamics. Consequently, the internal stability of feedback systems utilizing only relative measurements depends solely on the agent dynamics.

In addition to providing a rigorous analysis of the internal stability of diffusively-coupled systems, we show how the analysis is readily applied to common extensions found in the literature.  In particular, we discuss more general symmetrically coupled multi-agent systems (i.e.\ not restricted to only diffusive coupling), asymmetric coupling (i.e.\ MASs over directed graphs), unstable systems with no closed right-half plane poles, and MASs over time-varying networks. Numerous examples are also provided along the way to illustrate the main results.

The paper is organized as follows. The problem is set up in Section~\ref{sec:setup} and the main result is presented in Section~\ref{sec:stability}, with several generalizations discussed in \S\ref{sec:ext}. Section~\ref{sec:fd} addresses the case of finite-dimensional agents, reformulating the main result in a more transparent form and revealing the underlying reason for the reported behavior.  Concluding remarks are provided in Section~\ref{sec:concr}. Two appendices collect definitions and technical results about coprime factorizations over $H_\infty$ and poles and zero directions of multivariable real-rational transfer functions.

\subsubsection*{Notation}

The sets of integer, real, and complex numbers are $\mathbb Z$, $\mathbb R$, and $\mathbb C$, respectively, with subsets $\mathbb N_\nu\coloneq\{i\in\mathbb Z\mid1\le i\le\nu\}$, $\mathbb C_0\coloneq\{s\in\mathbb C\mid\real s>0\}$, and $\bar{\mathbb C}_0\coloneq\{s\in\mathbb C\mid\real s\ge0\}$. By $I_\nu$ and $\mathbb1_\nu$ we denote the $\nu\times\nu$ identity matrix and $\nu$-dimensional vector of ones, respectively. When the dimension is immaterial or clear from context, we use $I$ and $\mathbb1$. The complex-conjugate transpose of a matrix $A$ is denoted by $A^\top$, the set of all its eigenvalues by $\spec(A)$, and its minimal singular value by $\sigmin A$. The notation $\diag\{A_i\}$ stands for a block-diagonal matrix with diagonal elements $A_i$. The image (range) and kernel (null) spaces of a matrix $A$ are notated $\im A$ and $\ker A$, respectively. Given two matrices $A$ and $B$, $A\otimes B$ denotes their Kronecker product.

By the stability of a system $G$ we understand its $L_2$-stability. It is known \cite[\S A.6.3]{CZw:20} that a $p\times m$ LTI system is causal and stable iff its transfer function $G(s)$ belongs to $H_\infty^{p\times m}$, which is the space of holomorphic and bounded functions $\mathbb C_0\to\mathbb C^{p\times m}$ (we write $H_\infty$ when the dimensions are clear). Given a real-rational transfer function $G(s)$, its McMillan degree is denoted by $\deg(G)$. By $\nrank G(s)$ we understand the normal rank of a function $G(s)$.

A digraph $\mathcal G=(\mathcal V,\mathcal E)$ consists of a vertex set $\mathcal V$ and an edge set $\mathcal E\subset\mathcal V\times\mathcal V$, see \cite{GR:01} for more details. The (oriented) incidence matrix of $\mathcal G$ is denoted by $E_\mathcal G$ or simply $E$ when the association with a concrete graph is clear. It is a $\abs{\mathcal V}\times\abs{\mathcal E}$ matrix, whose $(i,j)$ entry is
\[
 [E_\mathcal G]_{ij}=
  \begin{cases}
   1 & \text{if vertex $i$ is the head of edge $j$} \\
  -1 & \text{if vertex $i$ is the tail of edge $j$} \\
   0 & \text{if vertex $i$ does not belong to edge $j$}
  \end{cases}.
\]
Note that the construction of the incidence matrix implies that $\mathbb1^\top E_\mathcal G=0$ for every $\mathcal G$.

\section{Problem formulation} \label{sec:setup}

Consider $\nu$ continuous-time LTI agents $P_i$, each with $m$ inputs and $p$ outputs, who interact over a graph $\mathcal G$ with $\nu$ nodes and $\mu$ edges. In this formalism, agents $i$ and $j$ are neighbors, in the sense defined in the Introduction, if they are incident to the same edge.

\begin{figure}[t]
 \centering\includegraphics[scale=.667]{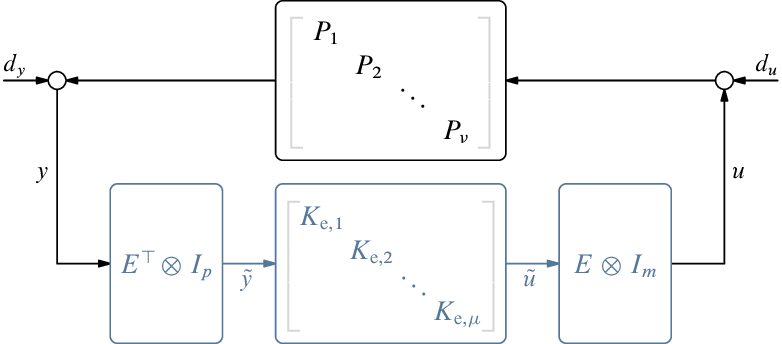}
 \caption{Diffusively-coupled feedback setup ($E$ is the incidence matrix of the connectivity graph $\mathcal G$)} \label{fig:BD_cl}
\end{figure}

A general diffusively-coupled MAS originated in \cite{Ar:07}, also known as the canonical cooperative control structure \cite[Ch.\,9]{B:22}, is presented in Fig.\,\ref{fig:BD_cl}. It comprises the block-diagonal aggregate plant $P\coloneq\diag\{P_i\}$ with $\nu$ blocks, a block-diagonal edge controller $K_\text e\coloneq\diag\{K_{\text e,j}\}$ with $\mu$ blocks, and pre- and post-processing based on the incidence matrix $E$ associated with $\mathcal G$. To describe the logic of this setup we may disregard the exogenous signals $d_y$ and $d_u$ for the time being. The overall controller $K:y\mapsto u$ is thus defined as
\begin{equation} \label{eq:K}
 K\coloneq(E\otimes I_m)K_\text e(E^{\top\!}\otimes I_p).
\end{equation}

We now discuss how the controller $K$ processes signals.
\begin{itemize}
\item The $(\nu p)$-dimensional aggregate output of the agents, $y$, is first processed by the transpose of the incidence matrix to produce a $(\mu p)$-dimensional vector $\tilde y=(E^\top\otimes I_p)y$ representing the relative outputs of neighbouring agents.
\item Each component of $\tilde y$, which is the relative measured coordinate along one edge, is then processed independently by an edge controller $K_{\text e,j}$, to produce a $(\mu m)$-dimensional ``edge correction'' signal $\tilde u$.
\item The $(\nu m)$-dimensional aggregate control signal $u$ is then produced by processing all $\tilde u_j$ by the incidence matrix, which sums up edge corrections for all edges connected to the corresponded node.
\end{itemize}
For example, if $\mathcal G$ is an undirected star graph on three nodes with node~3 as its center, then we can choose
\[
 E=\mmatrix{1&0\\0&1\\-1&-1},
\]
in which case
\[
 \tilde y=\mmatrix{y_1-y_3\\y_2-y_3}
 \quad\text{and}\quad
 u=\mmatrix{\tilde u_1\\\tilde u_2\\-\tilde u_1-\tilde u_2}.
\]
The consensus protocol \eqref{eq:conprot} corresponds to the choice $K_\text e=-I$ in this case, as well as for any other choice of $\mathcal G$ and $\nu$.

Now consider the exogenous signals $d_u$ and $d_y$, which we refer to as disturbances.  On the physical level they represent inevitable effects of the outside world on the controlled plant (agents). These signals are supposed to be bounded and independent of the signals generated by the controlled system. We introduce disturbances to define the notion of the internal stability for the system in Fig.\,\ref{fig:BD_cl}, which is the focus point of this paper. Specifically, we say that this system is \emph{internally stable} if the $2\times2$ operator connecting exogenous signals $d_u$ and $d_y$ with internal signals $u$ and $y$, i.e.\
\begin{equation} \label{eq:T4}
 T_4:(d_y,d_u)\mapsto(y,u)
\end{equation}
is well defined and stable, see \cite[Sec.\,4]{GSm:93}.

The general question of interest in this paper is \emph{under what conditions on the agents $P_i$ are there causal edge controllers $K_{\text e,j}$ internally stabilizing the diffusively-coupled system in Fig.\,\ref{fig:BD_cl}?} Note that the existence of edge controllers rendering the closed-loop operator well defined is obvious, just take $K_{\text e,j}=0$ for all $j$. We shall thus focus on the stability of $T_4$.

Addressing the stability question in the most general, nonlinear and time-varying, case might be overly technical. We thus limit our attention to the class of LTI plants and edge controllers, whose transfer functions belong to the quotient field of $H_\infty$, see \cite[\S A.7.1]{CZw:20}, which is a sufficiently general class. We further assume that
\begin{assumptions}
\item\label{ass:Picf} there are right coprime $M_i,N_i\in H_\infty$ and left coprime $\tilde M_i,\tilde N_i\in H_\infty$ such that $P_i=N_iM_i^{-1}=\tilde M_i^{-1}\tilde N_i$ for all $i$,
\end{assumptions}
where coprimeness is understood as the existence of Bézout coefficients in $H_\infty$, see Appendix \ref{sec:coprf}. The representation of $P_i$ above is known as its coprime factorization. We hereafter refer to the transfer functions $M_i(s)$ and $\tilde M_i(s)$ as the right and left denominators of $P_i$, respectively, and the transfer functions $N_i(s)$ and $\tilde N_i(s)$ as its right and left numerator.  Assumption \Ass{\ref{ass:Picf}} is practically nonrestrictive. It holds for all finite-dimensional agents with proper transfer functions and is equivalent to the stabilizability of $P_i$ by feedback for agents with transfer functions from the quotient field of $H_\infty$ \cite{Sm:89}. Thus, if an agent fails to satisfy \Ass{\ref{ass:Picf}}, we cannot expect any MAS that includes it to be stabilizable by diffusive coupling.

\begin{remark} \label{rem:Fig2vs3}%
 We choose the application points of exogenous disturbances for the internal stability analysis to be at the points where the agents, $P$, are connected with the controller $K$ defined in \eqref{eq:K}. In this choice we follow the physical nature of the interconnection in Fig.\,\ref{fig:BD_cl} and think of separating the blocks $E\otimes I$ and $E^{\top\!}\otimes I$ in the controller as merely a way to streamline the choice of the design parameters, which are the edge controllers in $K_\text e$.
 \begin{figure}[t]
  \centering\includegraphics[scale=.667]{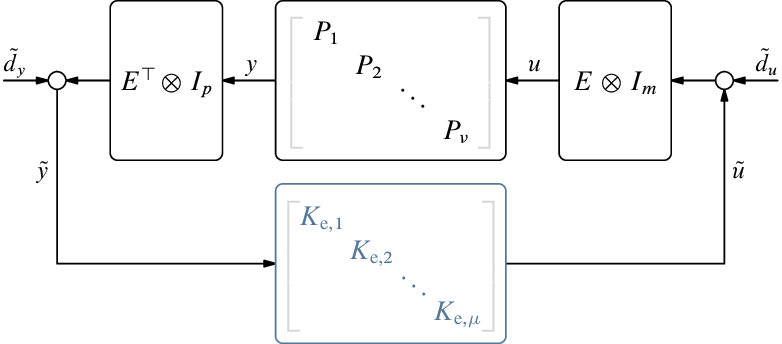}
  \caption{Diffusively-coupled feedback setup as edge stabilization} \label{fig:BD_es}
 \end{figure}
 An alternative viewpoint is presented in Fig.\,\ref{fig:BD_es}, where all \emph{fixed} parts are regarded as the controlled plant,
 \begin{equation} \label{eq:Pe}
  P_\text e\coloneq(E^{\top\!}\otimes I_p)P(E\otimes I_m),
 \end{equation}
 much inline with the generalized plant philosophy \cite[Sec.\,3.8]{SkP:05}, see e.g.\ \cite[Fig.\,6]{ZM:11} or \cite[E9.6]{B:22}. A natural definition of internal stability for it shall be based on the exogenous inputs $\tilde d_y$ and $\tilde d_u$, entering before and after the edge controller $K_\text e$. This would change the results, see Remark~\ref{rem:stcFig3} at the end of \S\ref{sec:calcc}.  Still, we believe that the configuration in Fig.\,\ref{fig:BD_cl} is the right way to address the internal stability of MASs. After all, it is the agents who interact with the environment.
\end{remark}

%%%%%%%%%%%%%%%%%%%%%%%%%%%%%%%%%%%%%%%%%%%%%%%%%%%%%%%%%%%%%%%%%%%%%%%%%%%%%%%%%%%%%

\section{The main result} \label{sec:stability}

The main technical result of this work, whose proof is postponed to \S\ref{sec:pfmr}, is formulated as follows.
\begin{theorem} \label{th:mainr}%
 No LTI $K_{\text e,j}$ can internally stabilize the diffusively-coupled system in Fig.\,\ref{fig:BD_cl} if there is $\lambda\in\bar{\mathbb C}_0$, common to all agents, such that
 \begin{subequations} \label{eq:thmCond}
 \begin{gather}
  \bigcap_{i=1}^\nu\ker\,[M_i(\lambda)]^\top\ne\{0\} \label{eq:thmCond:rcf} \\
 \shortintertext{or}
  \bigcap_{i=1}^\nu\ker\tilde M_i(\lambda)\ne\{0\}. \label{eq:thmCond:lcf}
 \end{gather}
 \end{subequations}
 where $M_i$ and $\tilde M_i$ are denominators in the coprime factorizations of $P_i$
 under \Ass{\ref{ass:Picf}}.
\end{theorem}

Theorem~\ref{th:mainr}, formulated in terms of coprime factors of agents, might appear somewhat abstract and technical. This is a consequence of considering a fairly general class of LTI agents under the mild assumption \Ass{\ref{ass:Picf}}. We show in the next section that if the class of admissible agents is limited to finite-dimensional ones, then more insightful statements can be provided. Nevertheless, the formulation in Theorem~\ref{th:mainr} becomes substantially more intuitive in some frequently studied special cases.

The first of them is the case of homogeneous agents, which is perhaps the best studied situation.
\begin{corollary} \label{cor:unifa}%
 If the agents are homogeneous, i.e.\:$P_i=P_0$ for all $i\in\mathbb N_\nu$, and $P_0(s)$ has at least one pole in\/ $\bar{\mathbb C}_0$, then no LTI $K_{\text e,j}$ can internally stabilize the system in Fig.\,\ref{fig:BD_cl}.
\end{corollary}
\begin{proof}
 By Lemma~\ref{lem:GpMz}, if $\lambda\in\bar{\mathbb C}_0$ is a pole of $P_0(s)$, then both $M_0(\lambda)$ and $\tilde M_0(\lambda)$ are singular, whence the result follows. \qed
\end{proof}
This result readily applies to the problem studied in \S\ref{sec:intro:mot}. The
agents in \eqref{eq:distcons} are homogeneous and $P_0(s)=1/s$ has an unstable pole at
the origin. Corollary~\ref{cor:unifa} then agrees with the conclusion of
\S\ref{sec:intro:mot} that the closed-loop system is not internally stable.

Another particular case for which the formulation is simplified is a MAS with SISO agents.
\begin{corollary} \label{cor:sisoa}%
 If the agents are SISO and all have a pole at the same\/ $\lambda\in\bar{\mathbb C}_0$, regardless of multiplicities, then no LTI $K_{\text e,j}$ can internally stabilize the diffusively-coupled system in Fig.\,\ref{fig:BD_cl}.
\end{corollary}
\begin{proof}
 By Lemma~\ref{lem:GpMz}, in this case $M_i(\lambda)=\tilde M_i(\lambda)=0$ for all $i\in\mathbb N_\nu$, whence the result follows. \qed
\end{proof}

A consequence of Corollary~\ref{cor:sisoa} is that the consensus protocol, as well as any other diffusively-coupled control laws, cannot internally stabilize a group of SISO agents if all of them contain an integral action. This result is reminiscent of that by \cite{WSA:11} that states that a common internal model is a necessary condition for a diffusively-coupled system to synchronize their state trajectories. It highlights a contradiction or trade-off of sorts, where on the one hand, a common pole at the origin among agents is required for synchronization, and on the other hand, this common (unstable) pole is precisely the cause for lack of internal stability.

\subsection{Proof of Theorem~\ref{th:mainr}} \label{sec:pfmr}

We are now prepared to prove Theorem~\ref{th:mainr}. Only the statement about the right coprime factor, i.e.\ \eqref{eq:thmCond:rcf}, is proved. The proof of \eqref{eq:thmCond:lcf} follows by dual arguments.

The proof requires a technical result of \citet{F:68tams}, known as the matrix corona theorem, see also the proof of \cite[Prop.\,11]{GSm:93} for a closer formulation.
\begin{lemma} \label{lem:GinvpGz}%
 If $G\in H_\infty^{n\times n}$, then
 \[
  G^{-1}\in H_\infty\iff\textstyle\inf_{s\in\bar{\mathbb C}_0}\sigmin{G(s)}>0.
 \]
\end{lemma}

It is readily seen that $M_P\coloneq\diag\{M_i\}$ and $N_P\coloneq\diag\{N_i\}$ are right coprime factors of $P=\diag\{P_i\}$. Because any internally stabilizing $K$ in \eqref{eq:K} is in effect stabilized by the plant, we only need to consider edge controllers for which $K$ admits coprime factorizations over $H_\infty$. So let $K=N_KM_K^{-1}$ for right coprime $M_K,N_K\in H_\infty$. By \eqref{eq:K},
\[
 N_K(s)=(E\otimes I_m)K_\text e(s)(E^{\top\!}\otimes I_p)M_K(s).
\]
Because $\mathbb1^\top E=0$, we have that $(\mathbb1^\top\otimes I_m)(E\otimes I_m)=0$ as well and, hence, $(\mathbb1^\top\otimes I_m)N_K(s)=0$ for all $s$ at which $K_\text e(s)$ is finite. But $K_\text e(s)$ is in the quotient field of $H_\infty$, meaning that the denominators of its entries are holomorphic in $\mathbb C_0$ and, by \cite[Thm.\,10.18]{R:87}, may have at most countable number of isolated zeros. As such, we can always find a region in $\mathbb C_0$ in which $(\mathbb1^\top\otimes I_m)N_K(s)=0$. But the latter implies that
\[
 (\mathbb1^{\top\!}\otimes I_m)N_K=0,
\]
by the same \cite[Thm.\,10.18]{R:87}.

Now, return to the system in Fig.\,\ref{fig:BD_cl}. It is readily verified that the closed-loop system $T_4$ in \eqref{eq:T4} reads
\begin{equation} \label{eq:gof}
 T_4=\mmatrix{I\\K}(I-PK)^{-1}\mmatrix{I&P}\eqcolon\mmatrix{S&T_\text d\\T_\text c&T},
\end{equation}
where the blocks of $T_4$ are the four fundamental closed-loop transfer functions. Straightforward algebra yields that
\begin{equation} \label{eq:gofrcf}
 T_4=\mmatrix{M_K&0\\N_K&0}\smash{\mmatrix{M_K&-N_P\\-N_K&M_P}^{-1}}.
\end{equation}
This is a right coprime factorization of $T_4$, as attested by the Bézout equality (cf.\ \eqref{eq:BezE:r})
\begin{multline*}
 \mmatrix{\tilde M_P&\tilde N_P\\-Y_P&X_P}\mmatrix{M_K&-N_P\\-N_K&M_P}\\+\mmatrix{X_K-\tilde M_P&Y_K+\tilde N_P\\Y_P&X_P}\mmatrix{M_K&0\\N_K&0}=I,
\end{multline*}
where $\tilde M_P\coloneq\diag\{\tilde M_i\}$ and
$\tilde N_P\coloneq\diag\{\tilde N_i\}$. By Lemma~\ref{lem:GstM}, $T_4$ is stable if and only if
\begin{equation} \label{eq:Mcl}
 \mmatrix{M_K&-N_P\\-N_K&M_P}^{-1\!}\in H_\infty,
\end{equation}
or
\begin{equation} \label{eq:Minvstc}
 \inf_{s\in\bar{\mathbb C}_0}\sigmin{\mmatrix{M_K(s)&-N_P(s)\\-N_K(s)&M_P(s)}}>0
\end{equation}
by Lemma~\ref{lem:GinvpGz}. But \eqref{eq:thmCond:rcf} implies that there is $v\ne0$ such that $v^\top M_i(\lambda)=0$ for all $i$ or, equivalently, $(\mathbb1\otimes v)^\top M_P(\lambda)=0$. Taking into account that $(\mathbb1\otimes v)^\top N_K=v^\top(\mathbb1\otimes I_m)^\top N_K=0$, we end up with
\begin{equation} \label{eq:cokerMcl}
 \mmatrix{0&(\mathbb1\otimes v)^\top}\mmatrix{M_K(\lambda)&-N_P(\lambda)\\-N_K(\lambda)&M_P(\lambda)}=0,
\end{equation}
which violates \eqref{eq:Minvstc}. We thus have that if \eqref{eq:thmCond:rcf} holds, then there is no $K_\text e$ that internally stabilizes the system in Fig.~\ref{fig:BD_cl}. \qed

\subsection{Generalizations} \label{sec:ext}

Some possible generalizations of the result of Theorem~\ref{th:mainr} are outlined below.

\subsubsection{Asymmetric coupling}

Some MAS problems consider a \emph{directed} interaction graph, making the notion of neighboring agents asymmetric. Controllers under such constrains are no longer diffusive in the sense discussed in Section~\ref{sec:setup}. Still, a variant of Theorem~\ref{th:mainr} may apply.

For example, let an edge going from node $i$ to node $j$ indicate that the $i$th agent has access to $y_i-y_j$. The existence of the edge $(i,j)$ does not imply that there is also the edge $(j,i)$. It is evident that the controller outlined in Fig.\,\ref{fig:BD_cl} and \eqref{eq:K} can no longer provide an appropriate distributed controller since, as discussed in Section~\ref{sec:setup}, it sums up all the edge correction terms connected to each corresponding node. Nevertheless, several notable MAS control architectures over directed graphs still admit a decomposition similar to that of \eqref{eq:K}. Consider again the classic consensus protocol. It can be adapted to accommodate directed graphs by replacing the symmetric Laplacian, $L=EE^\top$, with a directed counterpart such as the out-degree Laplacian $L_\text{out}$ \cite[Sec.\,7.3]{B:22}. By defining an auxiliary matrix,
\[
 [B_\text{out}]_{ij}=
  \begin{cases}
   1 & \text{if vertex $i$ is the head of edge $j$} \\
   0 & \text{otherwise}
  \end{cases},
\]
the directed out-degree Laplacian can be represented by the product $L_\text{out}=B_\text{out}E^\top$. This suggests that a controller of the form
\begin{equation} \label{eq:Kout}
 K_\text{out}\coloneq(B_\text{out}\otimes I_m)K_\text e(E^{\top\!}\otimes I_p),
\end{equation}
can be used to represent various control laws over directed graphs. For example setting $K_\text e=-I$ results in the aforementioned directed consensus protocol, while picking $K_\text e=I_\nu\otimes\tilde K$ for some gain $\tilde K$ yields the synchronizing controllers discussed in \cite[Sec.\,8.4]{B:22}.

The controller structure in \eqref{eq:Kout} mirrors that in \eqref{eq:K}. If \eqref{eq:thmCond:lcf} holds, then the proof of Theorem~\ref{th:mainr} applies verbatim to any MAS controlled by it. However, this is not the case for \eqref{eq:thmCond:rcf}, implying that some systems may be stabilizable only if the graph is directed, as illustrated in the following example.
\begin{example} \label{ex:Dir3}%
 Consider a system of $\nu=3$ first-order agents
 \[
  P_1(s)=\mmatrix{1&0\\1/s&1}
  \quad\text{and}\quad
  P_2(s)=P_3(s)=\mmatrix{1/s&0\\1&1}.
 \]
 Assume that their connectivity is represented by the directed cycle graph, which has three directed edges $(1,2)$, $(2,3)$, and $(3,1)$. This system can be described by \eqref{eq:Kout} with
 \[
  E=\mmatrix{1&0&-1\\-1&1&0\\0&-1&1},\quad B_\text{out}=I_3,
 \]
 and arbitrary block-diagonal edge controllers. It is then a matter of standard algebra to verify that these plants admit denominators
 \[
  \tilde M_1(s)=\mmatrix{1&0\\0&s/(s+1)}
 \]
 and
 \[
  \tilde M_2(s)=\tilde M_3(s)=\mmatrix{s/(s+1)&0\\0&1}=M_i(s),\quad\forall i\in\mathbb N_3.
 \]
 Hence, condition \eqref{eq:thmCond:rcf} holds for $\lambda=0$, whereas condition \eqref{eq:thmCond:lcf} holds for no $\lambda$. Thus, if the interconnection graph was undirected, then Theorem~\ref{th:mainr} would rule out the existence of internally stabilizing edge controllers. But in the directed case in form \eqref{eq:Kout} with the identity $B_\text{out}$ what matters is only \eqref{eq:thmCond:lcf}. Hence, we cannot rule out the existence of an internally stabilizing controller. And indeed, it can be verified that
 \[
  K_\text e(s)=-\diag\left\{\mmatrix{1&1\\5/3&1},I_2,I_2\right\}
 \]
 results in an internally stable interconnection, with the closed-loop poles in $\{-1/2,-3/4,-1\}$.
\end{example}
Of course, following a similar procedure we may define the analogous $K_\text{in}$ (corresponding for example to the in-degree directed consensus protocol) and consider only condition \eqref{eq:thmCond:rcf}, then again the proof holds unchanged.

\begin{remark}
 The stabilizability of control architectures over directed graphs may nevertheless still require checking both conditions of Theorem~\ref{th:mainr}. This thesis is based on an interpretation of the edge controller \eqref{eq:Kout} as (dynamic) edge weights of the directed graph. A directed graph is called \emph{weight balanced} if the accumulated weights of incoming and outgoing edges are equal for each node. It is known \cite[Thm.\,3.17]{ME:10} that the consensus protocol for integrator agents can reach an average agreement, i.e.\ $x_i(t)\to(1/\nu)\mathbb1'x(0)$ for all $i$, iff the underlying digraph is weight balanced and weakly connected. A key property to prove this result is that the Laplacian of a weight-balanced digraph, $L_\text{out}$, satisfies $\ker L_\text{out}=\ker L_\text{out}^\top=\im\mathbb1$. Viewed within the context of Theorem~\ref{th:mainr}, this implies that if edge controllers in \eqref{eq:Kout} are chosen such that digraph is weight balanced, then both conditions of \eqref{eq:thmCond} must be checked anyway.
\end{remark}

\subsubsection{Arbitrary symmetric coupling}

The result of Theorem~\ref{th:mainr} still holds if the incidence matrix is replaced with a different coupling matrix, say $F\in\mathbb R^{\mu\times\nu}$, as long as there is a vector $0\ne v\in\mathbb R^\mu$ such that $v^\top F=0$. Such generalizations of a MAS were recently discussed in \cite{BChZ:21}, but are also included in works considering, for example, distributed function calculation in MAS \cite{SH:08}.

\subsubsection{Unstable systems with no poles in $\bar{\mathbb C}_0$}

It might happen that $P_i\not\in H_\infty$ not because of poles, or other singularities, in $\bar{\mathbb C}_0$.  For example, $P_i(s)=s/(s+1+s\mathrm e^{-s})$ has no singularities in $\bar{\mathbb C}_0$, but nonetheless does not belong to $H_\infty$, see \cite{PB:04}. The proof still applies in this case, and all we need is to replace \eqref{eq:thmCond} with the assumption that there is a sequence $\{\lambda_j\}$ in $\mathbb C_0$ such that $\inf_{\{\lambda_j\}}v^\top M_i(\lambda_j)=0$, or its dual version, holds for all $i\in\mathbb N_\nu$ and some $v\ne0$.

\subsubsection{Time-varying $K$}

The main result also extends to the case of time-varying controllers. This is particularly relevant for varying interconnection topologies, i.e.\ those where $E_{\mathcal G(t)}=E(t)$ is the incidence matrix of the time-varying graph $\mathcal G(t)$. Still, the condition $\mathbb1^\top E(t)$ holds for any topology, rendering the denominator in \eqref{eq:gofrcf} not stably invertible. We can then use \cite[Theorem~(i)]{V:88} to show that under no choice of $K_\text e$ the system is stabilizable, at least in the finite-dimensional case, whenever either one of the conditions in \eqref{eq:thmCond} holds.

\section{Finite-dimensional agents} \label{sec:fd}

If the agents $P_i$ are finite dimensional, the result of the previous section can be reformulated in a more insightful way. This is due to the ultimate connection between stability and pole locations, as well as clear definitions of cancellations in this case. So we proceed with assuming that all transfer functions $P_i(s)$ are real rational and proper (\Ass{\ref{ass:Picf}} always holds then).

Let $\pdiri G\lambda$ and $\pdiro G\lambda$ denote input and output direction of a pole $\lambda$ in $G(s)$, see \appendixname\ref{sec:pzdir} for details and other related definitions. The result below reformulates the conditions of Theorem~\ref{th:mainr} via pole directions of agents.
\begin{proposition} \label{pr:fdmain}%
 If $P_i(s)$ are real rational and proper, then \eqref{eq:thmCond:rcf} and \eqref{eq:thmCond:lcf} are equivalent to the existence of\/ $\lambda\in\bar{\mathbb C}_0$ such that
 \begin{subequations} \label{eq:thmCondir}
 \begin{gather}
  \bigcap_{i=1}^\nu\,\pdiri{P_i}\lambda\ne\{0\} \label{eq:thmCondir:i} \\
 \shortintertext{and}
  \bigcap_{i=1}^\nu\,\pdiro{P_i}\lambda\ne\{0\}, \label{eq:thmCondir:o}
 \end{gather}
 \end{subequations}
 respectively.
\end{proposition}
\begin{proof}
 Because $\lambda\in\bar{\mathbb C}_0$ is not a pole of $M_i(s)$, Lemma~\ref{lem:zerker} applies and \eqref{eq:thmCond:rcf} reads $\cap_{i=1}^\nu\,\zdiro{M_i}\lambda\ne\{0\}$. Then \eqref{eq:thmCondir:i} follows by Lemma~\ref{lem:GpMzDir}. The proof for \eqref{eq:thmCondir:o} is similar. \qed
\end{proof}

In other words, for the system in Fig.\,\ref{fig:BD_cl} to not be stabilizable, the agents should not only have a common unstable pole, but also a common nontrivial direction of such a pole. Directions are obviously matched in the homogeneous and SISO cases addressed in Corollaries~\ref{cor:unifa} and \ref{cor:sisoa}, respectively. But the MIMO heterogeneous case may be less trivial.
\begin{example}
 Consider a system with $\nu=2$ first-order agents
 \[
  P_1(s)=\mmatrix{1/s&0\\0&1}
  \quad\text{and}\quad
  P_2(s)=\mmatrix{1\\\alpha}\frac1s\mmatrix{1&\beta}.
 \]
 Directions of their pole at the origin are
 \begin{gather*}
  \pdiri{P_1}0=\pdiro{P_1}0=\im\mmatrix{1\\0}, \\
  \pdiri{P_2}0=\im\mmatrix{1\\\beta},\quad\text{and}\quad\pdiro{P_2}0=\im\mmatrix{1\\\alpha}.
 \end{gather*}
 There are nontrivial intersections between input and output directions of the agents if and only if $\beta=0$ and $\alpha=0$, respectively. The incidence matrix is $E=\smat{1\\-1}$ in this case. Choose the edge controller (there is only one edge in this example) as
 \[
  K_\text e(s)=\mmatrix{(\alpha-\beta)\beta&-\alpha\\\beta&0}.
 \]
 The closed-loop characteristic polynomial, understood as the lowest common denominator of elements of $T_4(s)$ in \eqref{eq:gof}, is then $(s+\alpha^2)(s+\beta^2)$. Thus, the closed-loop system is stable unless $\alpha=0$ or $\beta=0$, which agrees with \eqref{eq:thmCondir}.
\end{example}

Also worth emphasizing is that conditions \eqref{eq:thmCondir:i} and \eqref{eq:thmCondir:o} might not be equivalent for MIMO agents, as illustrated by the example below.
\begin{example} \label{ex:P12neP21}%
 Return to the system studied in Example~\ref{ex:Dir3}. Directions associated with the (unstable) pole at the origin are
 \begin{gather*}
  \pdiri{P_i}0=\im\mmatrix{1\\0},\quad\forall i\in\mathbb N_3 \\
 \shortintertext{but}
  \pdiro{P_1}0=\im\mmatrix{0\\1}\ne\im\mmatrix{1\\0}=\pdiro{P_2}0.
 \end{gather*}
 Thus, in this case \eqref{eq:thmCondir:i} holds, whereas \eqref{eq:thmCondir:o} does not. This agrees with what we saw in Example~\ref{ex:Dir3} with respect to conditions \eqref{eq:thmCond}.
\end{example}

Another outcome of the finite dimensionality is that the formulation of Corollary~\ref{cor:unifa} can be strengthened to an ``if and only if'' statement.
\begin{corollary} \label{cor:unifda}%
 If the agents are homogeneous, i.e.\:$P_i=P_0$ for all $i\in\mathbb N_\nu$, and $P_0(s)$ is real rational and proper, then an LTI $K_{\text e,j}$ can internally stabilize the diffusively-coupled system in Fig.\,\ref{fig:BD_cl} if and only if $P_0$ is stable.
\end{corollary}
\begin{proof}
 If $P_0$ is unstable, then it has a pole in $\bar{\mathbb C}_0$ and Corollary~\ref{cor:unifa} applies. If $P_0$ is stable, $K_\text e=0$ does the job. \qed
\end{proof}

One should be careful not to conclude from the proof of Corollary~\ref{cor:unifda} that only $K_\text e=0$ can be used to guarantee internal stability. The case of $K_\text e=0$ effectively decouples all the agents leading only to a ``trivial'' coordination (i.e.\ all agents converge to the origin). One can design edge controllers with additional external inputs to drive the relative states $\tilde y$ to non-trivial solutions using the methods, for example, described in \cite{SZ:17}. For non-trivial agreement among the agents, the use of an unstable edge controller is possible provided that an appropriately defined external input is fed into the system at the point $d_y$ in Fig. \ref{fig:BD_cl}.

\subsection{Diffusive control laws and unstable cancellations} \label{sec:calcc}

The formulation of Proposition~\ref{pr:fdmain} is more intuitive than that of Theorem~\ref{th:mainr}. Still, neither of them explains \emph{why} no edge controller can stabilize the system in Fig.\,\ref{fig:BD_cl} if agents share common unstable dynamics, directions counted. In this part we aim at offering explanations. We argue that a key property to this end is intrinsic \emph{unstable cancellations} between the plant and the controller.

The cascade (series) interconnection $G_2G_1$ has cancellations if $\deg(G_2G_1)<\deg(G_1)+\deg(G_2)$. In other words, cancellations mean that some parts of the dynamics (modes) of either factor disappear in the cascade. Specifically, we say that a pole of $G_1(s)$ and/or $G_2(s)$ is \emph{canceled} if its multiplicity in $G_2(s)G_1(s)$ is smaller than the sum of its multiplicities in $G_1(s)$ and $G_2(s)$.  Cancellations in the SISO case are always caused by the presence of zeros of $G_1(s)$ at the locations of poles of $G_2(s)$, or vice versa. As such, they are termed \emph{pole-zero cancellations}. The situation is more complex in the MIMO case. For example, let
\[
 G_1(s)=\frac1s\mmatrix{1&0\\0&1}\quad\text{and}\quad G_2(s)=\mmatrix{1&-1\\-1&1},
\]
with $\deg(G_1)=2$ (two poles at the origin) and $\deg(G_2)=0$ (no poles). The system $G_2$ is static and thus has no zeros either. Nevertheless, the transfer function
\[
 G_2(s)G_1(s)=\frac1s\mmatrix{1&-1\\-1&1}
\]
is first order, meaning that one of the poles of $G_1(s)$ is canceled. Such cancellations, brought on by the normal rank deficiency of $G_2(s)$, are a lesser-known phenomenon.

The result below, proved in \S\ref{sec:pfcanc}, states that such cancellations are present between the plant and the controller in Fig.\,\ref{fig:BD_cl} whenever the conditions of Proposition~\ref{pr:fdmain} hold.
\begin{proposition} \label{pr:pcanc}%
 Let $P(s)$ and $K_\text e(s)$ be real rational and proper and let $\lambda\in\bar{\mathbb C}_0$ be a pole of $P(s)$.
 \begin{itemize}
 \item[i)] If \eqref{eq:thmCondir:i} holds, then $\lambda$ is canceled in $P(s)K(s)$.
 \item[ii)] If \eqref{eq:thmCondir:o} holds, then $\lambda$ is canceled in $K(s)P(s)$.
 \end{itemize}
\end{proposition}

Unstable pole-zero cancellations between a plant and a controller are a consensual taboo in feedback control. Textbooks treat them as a kind of a cardinal sin, which shall be avoided at all costs. The reason is that canceled dynamics do not really disappear. For example, poles of a SISO plant $P(s)$ canceled by zeros of a controller $K(s)$ always show up in the closed-loop disturbance sensitivity $T_\text d(s)$, see \eqref{eq:gof}. This is the very reason to require internal stability. Unstable cancellations due to deficient normal rank are less common and less studied.  Nevertheless, they cause same repercussions. Namely, canceled dynamics shows up in at least one closed-loop relation, rendering the system prone to the effect of exogenous signals.

Assume, for example, that condition \eqref{eq:thmCondir:i}, or \eqref{eq:thmCond:rcf}, holds for some $\lambda\in\bar{\mathbb C}_0$. It follows from the proof of Theorem~\ref{th:mainr} that there is then $v\ne0$ such that \eqref{eq:cokerMcl} holds.  Therefore,
\[
 \mmatrix{0\\\mathbb1\otimes v}\in\zdiro{\mmatrix{M_K&-N_P\\-N_K&M_P}}\lambda=\pdiri{T_4}\lambda
\]
where the equality follows by Lemma~\ref{lem:GpMzDir} and the fact that the factors in \eqref{eq:gofrcf} are right coprime. By Lemma~\ref{lem:pdirm} and \eqref{eq:gof}
\[
 T_4(s)\mmatrix{0\\\mathbb1\otimes v}=\mmatrix{T_\text d(s)\\T(s)}(\mathbb1\otimes v)
\]
has an unstable pole at $s=\lambda$. In other words, there is a load disturbance $d_u$ in Fig.~\ref{fig:BD_cl} such that either $y$ or $u$ or both is unbounded. Likewise, it can be shown that if \eqref{eq:thmCondir:o} holds, then $\mmatrix{S&T_\text d}\not\in H_\infty$, i.e.\ $d_u$ or/and $d_y$ might cause an unbounded $y$. This explains why the consensus protocol in \S\ref{sec:intro:mot} has an unstable load disturbance response.

It can be shown that if the consensus discussed in \S\ref{sec:intro:mot} can be attained, then all components of $T_4$ but $T_\text d$ are stable, whereas $T_\text d(s)$ has a pole at the origin. This agrees with the situation in SISO pole-zero cancellations discussed above. However, $T_\text d$ is not necessarily unstable in a general MIMO case if either of the conditions in \eqref{eq:thmCondir} holds. The example below illustrates a different scenario.
\begin{example} \label{ex:T422us}%
 Consider a system with $\nu=2$ agents
 \[
  P_1(s)=\mmatrix{s/(s+1)&0\\1/s&1}
  \quad\text{and}\quad
  P_2(s)=\mmatrix{1/s&0\\1&s/(s+1)}
 \]
 (both are second order). In this case there is only one edge. Select
 \[
  K_\text e(s)=K_{\text e,1}(s)=-\frac13\mmatrix{1&0\\0&2/s}.
 \]
 It is then a matter of routine calculations to see that $S$, $T_\text d$, and $T_\text c$ are stable, each having $(s+2)(2s+1)^2(3s+1)$ as the lowest common denominator of its entries. However, $T(s)$ has a pole at the origin in addition, rendering the whole $T_4$ unstable.
\end{example}
Moreover, it may even happen that canceled dynamics of $P$ are not excited by the (load) disturbance $d_u$, but rather only by $d_y$.
\begin{example} \label{ex:T421us}%
 Consider a system with $\nu=2$ agents, yet again, now with the second order
 \[
  P_1(s)=\mmatrix{s/(s+1)&1/s\\0&1},
  \quad
  P_2(s)=\mmatrix{1/s&1\\0&s/(s+1)}
 \]
 and the edge controller from Example~\ref{ex:T422us}. It can be calculated that in this case $T$, $T_\text d$, and $T_\text c$ are stable, each having $(s+2)(2s+1)^2(3s+1)$ as the lowest common denominator of its entries. The sensitivity $S(s)$ has an additional pole at the origin. This implies that the responses to $d_u$ are all stable, whereas the response of $y$ to $d_y$ is unstable.
\end{example}

\begin{remark} \label{rem:stcFig3}%
 Stabilizability conditions for the setup in Fig.~\ref{fig:BD_es} would be substantially different from those in Theorem~\ref{th:mainr} or Proposition~\ref{pr:fdmain}. If we consider the class of LTI edge controllers $K_\text e$, then the stabilizability problem boils down to the question of existing decentralized fixed modes (DFMs) in $P_\text e$ defined by \eqref{eq:Pe}, see \cite[Sec.\,2.2]{DAM:20}. If controllers are allowed to be periodically time-varying, then even this condition is not restrictive \cite{AM:81}. However, this analysis has a snag in that the very construction of $P_\text e$ might have unstable cancellations. For example, return to the case of $\nu=3$ integrator agents with an indirect star interconnection graph discussed in Section~\ref{sec:setup}. In this case $P(s)=(1/s)I_3$ has three poles at the origin, whereas
 \[
  P_\text e(s)
   =\mmatrix{1&0&-1\\0&1&-1}\Bigl(\frac1sI_3\Bigr)\mmatrix{1&0\\0&1\\-1&-1}
   =\frac1s\mmatrix{2&1\\1&2}
 \]
 is a second-order transfer function. This $P_\text e$ is easily stabilizable by decentralized edge controllers, e.g.\ by $K_\text e=-I_2$. But this controller cannot see the canceled unstable mode, which remains a part of the closed-loop system.
\end{remark}

\subsection{Proof of Proposition~\ref{pr:pcanc}} \label{sec:pfcanc}

Bring in \emph{minimal} realizations
\[
 P_i(s)=\mmatrix[c|c]{A_i&B_i\\\hline C_i&D_i}
 \quad\text{and}\quad
 K(s)=\mmatrix[c|c]{A_K&B_K\\\hline C_K&D_K}
\]
so the realization
\[
 P(s)=\mmatrix[c|c]{A_P&B_P\\\hline C_P&D_P}\coloneq\mmatrix[c|c]{\diag\{A_i\}&\diag\{B_i\}\\\hline\diag\{C_i\}&\diag\{D_i\}}
\]
is also minimal. To prove the first item of the Proposition it is then sufficient to
show that $\lambda$ is an uncontrollable mode of
\[
 P(s)K(s)=\mmatrix[cc|c]{A_K&0&B_K\\B_PC_K&A_P&B_PD_K\\\hline D_PC_K&C_P&D_PD_K}.
\]
To this end, note that \eqref{eq:K} implies $(\mathbb1\otimes I)^\top\mmatrix{C_K&D_K}=0$
and condition \eqref{eq:thmCondir:i} is equivalent to the existence of
$0\ne v\in\mathbb C^m$ such that $v=B_i^\top\eta_i$ for some $\eta_i$ such that
$\eta_i^\top(\lambda I-A_i)=0$.  The latter is equivalent to the existence of
$\eta\ne0$ such that
\[
 \eta^\top(\lambda I-A_P)=0\quad\text{and}\quad\eta^\top B_P=(\mathbb1\otimes v)^\top
\]
for some $v\ne0$. Therefore,
\begin{multline*}
 \mmatrix{0&\eta^\top}\mmatrix{A_K-\lambda I&0&B_K\\B_PC_K&A_P-\lambda I&B_PD_K} \\
  =v^\top(\mathbb1\otimes I)^\top\mmatrix{C_K&0&D_K}=0
\end{multline*}
and the PBH test for the realization of $PK$ fails for the mode at $\lambda$, proving
the first item. The second item follows by similar arguments. \qed

%%%%%%%%%%%%%%%%%%%%%%%%%%%%%%%%%%%%%%%%%%%%%%%%%%%%%%%%%%%%%%%%%%%%%%%%%%%%%%%%%%%%%

\section{Concluding remarks} \label{sec:concr}

In this paper we have studied the internal stability of multi-agent systems controlled by diffusively coupled laws. We have argued that internal stability, with entry points of exogenous signals at the connections between the agents and the controller, is a vital property in multi-agent systems and have proved that it can never be attained if the agents share common unstable dynamics, directions counted. In particular, this class always includes the case of homogeneous unstable agents or heterogeneous SISO agents with a common unstable pole, like an integral action. We have shown that the underlying reason for the lack of stabilizability is intrinsic cancellations of aligned unstable dynamics of agents by the diffusive coupling mechanism.

An immediate outcome of the proposed analysis is that the uniformity must be broken in the control of unstable multi-agent systems. This is the underlying reason behind several of the different assumptions mentioned in Section~\ref{sec:intro}. Introducing a leader, or ``virtual'' agent, can potentially break the common instability, while permitting non-relative feedback either locally stabilize the agents or again, break the uniformity.

%%%%%%%%%%%%%%%%%%%%%%%%%%%%%%%%%%%%%%%%%%%%%%%%%%%%%%%%%%%%%%%%%%%%%%%%%%%%%%%%%%%%%

%%%%%%%%%%%%%%%%%%%%%%%%%%%%%%%%%%%%%%%%%%%%%%%%%%%%%%%%%%%%%%%%%%%%%%%%%%%%%%%%%%%%%

\appendix
\gdef\thesection{\Alph{section}}
\makeatletter
\renewcommand\@seccntformat[1]{\appendixname\csname the#1\endcsname.\hspace{0.5em}}
\makeatother

\section{Coprime factorizations over $\boldsymbol{H_\infty}$} \label{sec:coprf}

In this Appendix, basic coprime factorization results that are required in the paper are presented. A comprehensive exposition of the subject can be found in \cite{V:85}.

Functions $M\in H_\infty^{m\times m}$ and $N\in H_\infty^{p\times m}$ are said to be \emph{right coprime} if there are $X\in H_\infty^{m\times m}$ and $Y\in H_\infty^{m\times p}$ (Bézout coefficients) such that
\begin{subequations} \label{eq:BezE}
\begin{equation} \label{eq:BezE:r}
 XM+YN=I_m.
\end{equation}
Functions $\tilde M\in H_\infty^{p\times p}$ and $\tilde N\in H_\infty^{p\times m}$ are said to be \emph{left coprime} if there are $\tilde X\in H_\infty^{p\times p}$ and $\tilde Y\in H_\infty^{m\times p}$ such that
\begin{equation} \label{eq:BezE:l}
 \tilde M\tilde X+\tilde N\tilde Y=I_p.
\end{equation}
\end{subequations}
A transfer function $G(s)$ is said to have coprime factorizations over $H_\infty$ if there are right coprime $M_G,N_G\in H_\infty$ and left coprime $\tilde M_G,\tilde N_G\in H_\infty$, known as right and left coprime factors of $G$, respectively, such that
\begin{equation} \label{eq:Gcf}
 G=N_GM_G^{-1}=\tilde M_G^{-1}\tilde N_G.
\end{equation}
Coprime factors are unique up to post- or pre-multiplication by bi-stable transfer functions for right and left factors, respectively.

\begin{lemma} \label{lem:GstM}%
 If\/ $G(s)$ has coprime factorizations, then
 \[
  G\in H_\infty\iff M_G^{-1}\in H_\infty\iff\tilde M_G^{-1}\in H_\infty.
 \]
\end{lemma}
\begin{proof}
 The ``if'' part of the first equivalence relation is immediate from \eqref{eq:Gcf}. Its ``only if'' part follows from rewriting the Bézout equality \eqref{eq:BezE:r} as $M_G^{-1}=X_G+Y_GG$. The second relation follows by similar arguments. \qed
\end{proof}

\begin{lemma} \label{lem:GpMz}%
 Let $G(s)$ have coprime factorizations. If\/ $\lambda\in\bar{\mathbb C}_0$ is a pole of\/ $G(s)$, then $M_G(\lambda)$ and $\tilde M_G(\lambda)$ are singular.
\end{lemma}
\begin{proof}
 Because $\lambda\in\bar{\mathbb C}_0$, the singularity of $M_G(\lambda)$ or $\tilde M_G(\lambda)$ does not depend on concrete factorizations taken. If $M_G(\lambda)$ is nonsingular, then $N_G(\lambda)M_G(\lambda)^{-1}$ is bounded, which implies that $\lambda$ cannot be a pole of $G(s)$. The proof for $\tilde M_G$ is similar. \qed
\end{proof}

\section{Poles, zeros, and their directions} \label{sec:pzdir}

This \appendixname collects some definitions and facts on poles, zeros, and their directions for MIMO transfer functions. More details can be found in \cite{SkP:05}, although we use slightly different definitions of directions (subspaces, rather than vectors), in line with \cite{Mir:lcs}.

Let $G$ be a finite-dimensional LTI system having a proper transfer function $G(s)$. The system $G$ has a state-space realization
\begin{equation} \label{eq:ssr}
 G(s)=\mmatrix[c|c]{A&B\\\hline C&D}\coloneq D+C(sI-A)^{-1}B.
\end{equation}
The eigenvalues of $A$ are known as \emph{poles} of the realization \eqref{eq:ssr}. The set of all realization poles, multiplicities counted, coincides with that of the poles of the transfer function $G(s)$ if and only if the realization is minimal. \emph{Invariant zeros} of the realization \eqref{eq:ssr} are defined as the points $\lambda\in\mathbb C$ at which
\[
 \rank\mmatrix{A-\lambda I&B\\C&D}<\nrank\mmatrix{A-sI&B\\C&D}
\]
(the matrix polynomial of $s$ in the right-hand side is dubbed the Rosenbrock system matrix). The set of all invariant zeros comprises transmission zeros of the transfer function $G(s)$ and hidden modes of realization \eqref{eq:ssr}.

Poles and zeros have (spatial) directions for MIMO systems. Assume through the rest of this \appendixname that the realization in \eqref{eq:ssr} is minimal. By \emph{input} and \emph{output} directions of a realization pole $\lambda$ of \eqref{eq:ssr}, we understand the subspaces
\begin{subequations} \label{eq:pdir}
\begin{align}
 \pdiri G\lambda&\coloneq B^\top\ker(\lambda I-A)^\top\subset\mathbb C^m \label{eq:pdir:i} \\
 \shortintertext{and}
 \pdiro G\lambda&\coloneq C\ker(\lambda I-A)\subset\mathbb C^p, \label{eq:pdir:o}
\end{align}
\end{subequations}
respectively. If $\lambda$ is not a pole of $G(s)$, then both definitions in \eqref{eq:pdir} result in the trivial subspace $\{0\}$.
\begin{lemma} \label{lem:pdirm}%
 If\/ $\lambda\in\mathbb C$ is a pole of $G(s)$, then
 \begin{itemize}
 \item[i)] $\lambda$ is a pole of $G(s)v$ whenever $0\ne v\in\pdiri G\lambda$,
 \item[ii)] $\lambda$ is a pole of $v^\top G(s)$ whenever $0\ne v\in\pdiro G\lambda$.
 \end{itemize}
\end{lemma}
\begin{proof}
 Bring in a minimal realization of $G$ as in \eqref{eq:ssr}. If $(A,Bv)$ is controllable, then every eigenvalue of $A$ is a pole of $G(s)v$, by the observability of $(C,A)$. If $(A,Bv)$ is uncontrollable, without loss of generality we may assume that
 \[
  (A,B)=\left(\mmatrix{A_\text c&A_{12}\\0&A_\text{\=c}},\mmatrix{B_\text c\\B_\text{\=c}}\right)
 \]
 with controllable $(A_\text c,B_\text cv)$ and $B_\text{\=c}v=0$. In this case $\lambda$ is not a pole of $G(s)v$ iff $\lambda\not\in\spec(A_\text c)$. So assume that $\lambda\not\in\spec(A_\text c)$, which implies that $\lambda\in\spec(A_\text{\=c})$ and that
 \[
  B^\top\ker(\lambda I-A)^\top\subset\mmatrix{B_\text c^\top&B_\text{\=c}^\top}\im\mmatrix{0\\I}=\im B_\text{\=c}^\top.
 \]
 But then $v\in\pdiri G\lambda\implies v\in\im B_\text{\=c}^\top=(\ker B_\text{\=c})^\perp$, which contradicts the condition $B_\text{\=c}v=0$. Hence, $\lambda$ must be a pole of $G(s)v$. The second item follows by similar arguments. \qed
\end{proof}
Input and output directions of an invariant zero $\lambda$ are defined as%
\begin{subequations} \label{eq:zdir}
\begin{align}
 \zdiri G\lambda&\coloneq\mmatrix{0&I_m}\ker\mmatrix{A-\lambda I&B\\C&D}\subset\mathbb C^m \label{eq:zdir:i} \\
 \shortintertext{and}
 \zdiro G\lambda&\coloneq\mmatrix{0&I_p}\ker\mmatrix{A-\lambda I&B\\C&D}{\vphantom{\mmatrix{1\\[-5pt]1}}}^\top\!\subset\mathbb C^p, \label{eq:zdir:o}
\end{align}
\end{subequations}
respectively. With some abuse of notation we use the definitions in \eqref{eq:zdir} also if $\lambda$ is not an invariant zero of \eqref{eq:ssr}, but the normal rank of $G(s)$ is deficient. For example, in our notation
\[
 \zdiri{\mmatrix{1&-1\\-1&1}}\lambda=\zdiro{\mmatrix{1&-1\\-1&1}}\lambda=\im\mathbb1_2
\]
for all $\lambda\in\mathbb C$. In such situations directions are understood as normal null spaces.

\begin{lemma} \label{lem:zerker}%
 If $\lambda\not\in\spec(A)$, then it is an invariant zero of $G$ iff\/ $\rank G(\lambda)<\nrank G(s)$ and
 \[
  \zdiri G\lambda=\ker G(\lambda)
  \quad\text{and}\quad
  \zdiro G\lambda=\ker\,[G(\lambda)]^\top.
 \]
\end{lemma}
\begin{proof}
 Follows from the relations
 \begin{align*}
  \mmatrix{A-\lambda I&B\\C&D}
   &=\mmatrix{A-\lambda I&0\\C&G(\lambda)}\mmatrix{I&(A-\lambda I)^{-1}B\\0&I} \\
   &=\mmatrix{I&0\\C(A-\lambda I)^{-1}&I}\mmatrix{A-\lambda I&B\\0&G(\lambda)}
 \end{align*}
 and the assumed invertibility of $A-\lambda I$. \qed
\end{proof}

\begin{lemma} \label{lem:GpMzDir}%
 If $\lambda\in\bar{\mathbb C}_0$, then it is a pole of $G(s)$ if and only if it is a zero of the denominators $M_G(s)$ and $\tilde M_G(s)$ of its coprime factorizations.  Moreover,
 \[
  \pdiri G\lambda=\zdiro{M_G}\lambda\quad\!\!\text{and}\!\quad\pdiro G\lambda=\zdiri{\smash{\tilde M_G}}\lambda
 \]
 in this case.
\end{lemma}
\begin{proof}
 Follows by \cite[Prop.\,4.16]{Mir:lcs} and the fact that a pole of $G(s)$ in $\bar{\mathbb C}_0$ is a zero of all possible denominators. \qed
\end{proof}


\begin{thebibliography}{31}
\providecommand\url[1]{\texttt{#1}}
\providecommand\urlprefix{URL }
\expandafter\ifx\csname urlstyle\endcsname\relax
  \providecommand{\doi}[1]{doi:\discretionary{}{}{}#1}\else
  \providecommand{\doi}{doi:\discretionary{}{}{}\begingroup
  \urlstyle{rm}\Url}\fi

\bibitem[{Anderson and Moore(1981)}]{AM:81}
Anderson, B.D.O. and Moore, J.B. (1981).
\newblock Time-varying feedback laws for decentralized control.
\newblock \emph{IEEE Trans.\ Automat.\ Control}, 26(5), 1133--1139.

\bibitem[{Arcak(2007)}]{Ar:07}
Arcak, M. (2007).
\newblock Passivity as a design tool for group coordination.
\newblock \emph{IEEE Trans.\ Automat.\ Control}, 52(8), 1380--1390.

\bibitem[{Belabbas et~al.(2021)Belabbas, Chen, and Zelazo}]{BChZ:21}
Belabbas, M.A., Chen, X., and Zelazo, D. (2021).
\newblock On structural rank and resilience of sparsity patterns.
\newblock \emph{arXiv}.

\bibitem[{Bullo(2022)}]{B:22}
Bullo, F. (2022).
\newblock \emph{Lectures on Network Systems}.
\newblock Kindle Direct Publishing, 1.6 edition.
\newblock \urlprefix\url{http://motion.me.ucsb.edu/book-lns}.

\bibitem[{Bürger and {De Persis}(2015)}]{BDeP:15}
Bürger, M. and {De Persis}, C. (2015).
\newblock Dynamic coupling design for nonlinear output agreement and
  time-varying flow control.
\newblock \emph{Automatica}, 51, 210--222.

\bibitem[{Curtain and Zwart(2020)}]{CZw:20}
Curtain, R.F. and Zwart, H. (2020).
\newblock \emph{Introduction to Infinite-Dimensional Systems Theory: A
  State-Space Approach}.
\newblock Springer-Verlag, New York, NY.

\bibitem[{Davison et~al.(2020)Davison, Aghdam, and Miller}]{DAM:20}
Davison, E.J., Aghdam, A.G., and Miller, D.E. (2020).
\newblock \emph{Decentralized Control of Large-Scale Systems}.
\newblock Springer-Verlag, New York, NY.

\bibitem[{Ding(2015)}]{Ding:15}
Ding, Z. (2015).
\newblock Consensus disturbance rejection with disturbance observers.
\newblock \emph{IEEE Transactions on Industrial Electronics}, 62(9),
  5829--5837.

\bibitem[{Fuhrmann(1968)}]{F:68tams}
Fuhrmann, P.A. (1968).
\newblock On the corona theorem and its application to spectral problems in
  {Hilbert} space.
\newblock \emph{Trans.\ Amer.\ Math.\ Soc.}, 132(1), 55--66.

\bibitem[{Georgiou and Smith(1993)}]{GSm:93}
Georgiou, T.T. and Smith, M.C. (1993).
\newblock Graphs, causality and stabilizability: linear, shift-invariant
  systems on {$\mathcal{L}_2[0,\infty)$}.
\newblock \emph{Math.\ Control, Signals and Systems}, 6, 195--223.

\bibitem[{Godsil and Royle(2001)}]{GR:01}
Godsil, C.D. and Royle, G.F. (2001).
\newblock \emph{Algebraic Graph Theory}.
\newblock Springer.

\bibitem[{Khan et~al.(2009)Khan, Kar, and Moura}]{KKM:09}
Khan, U.A., Kar, S., and Moura, J.M.F. (2009).
\newblock Distributed sensor localization in random environments using minimal
  number of anchor nodes.
\newblock \emph{IEEE Transactions on Signal Processing}, 57(5), 2000--2016.

\bibitem[{Mesbahi and Egerstedt(2010)}]{ME:10}
Mesbahi, M. and Egerstedt, M. (2010).
\newblock \emph{Graph Theoretic Methods in Multiagent Networks}.
\newblock Princeton University Press, Princeton.

\bibitem[{Mirkin(2019)}]{Mir:lcs}
Mirkin, L. (2019).
\newblock {Linear Control Systems}.
\newblock course notes, Faculty of Mechanical Eng., Technion---IIT.
\newblock \urlprefix\url{http://leo.technion.ac.il/Courses/LCS/LCSnotes.pdf}.

\bibitem[{Mo and Guo(2019)}]{MG:19}
Mo, L. and Guo, S. (2019).
\newblock Consensus of linear multi-agent systems with persistent disturbances
  via distributed output feedback.
\newblock \emph{Journal of Systems Science and Complexity}, 32(3), 835--845.

\bibitem[{Olfati-Saber et~al.(2007)Olfati-Saber, Fax, and Murray}]{O-SFM:07}
Olfati-Saber, R., Fax, A., and Murray, R.M. (2007).
\newblock Consensus and cooperation in networked multi-agent systems.
\newblock \emph{Proc.\ IEEE}, 95(1), 215--233.

\bibitem[{Partington and Bonnet(2004)}]{PB:04}
Partington, J.R. and Bonnet, C. (2004).
\newblock {$H_\infty$} and {BIBO} stabilization of delay systems of neutral
  type.
\newblock \emph{Syst.\ Control Lett.}, 52(8), 283--288.

\bibitem[{Ren and Beard(2008)}]{RB:08}
Ren, W. and Beard, R.W. (2008).
\newblock \emph{Distributed Consensus in Multi-vehicle Cooperative Control:
  Theory and Applications}.
\newblock Springer-Verlag, London.

\bibitem[{Rudin(1987)}]{R:87}
Rudin, W. (1987).
\newblock \emph{Real and Complex Analysis}.
\newblock McGraw-Hill, New York, NY, 3rd edition.

\bibitem[{Sharf and Zelazo(2017)}]{SZ:17}
Sharf, M. and Zelazo, D. (2017).
\newblock A network optimization approach to cooperative control synthesis.
\newblock \emph{IEEE Control Syst.\ Lett.}, 1(1), 86--91.

\bibitem[{Skogestad and Postlethwaite(2005)}]{SkP:05}
Skogestad, S. and Postlethwaite, I. (2005).
\newblock \emph{Multivariable Feedback Control: Analysis and Design}.
\newblock John Wiley {\&} Sons, Chichester, 2nd edition.

\bibitem[{Smith(1989)}]{Sm:89}
Smith, M.C. (1989).
\newblock On stabilization and the existence of coprime factorizations.
\newblock \emph{IEEE Trans.\ Automat.\ Control}, 34(9), 1005--1007.

\bibitem[{Smith and Hadaegh(2005)}]{SH:05}
Smith, R.S. and Hadaegh, F.Y. (2005).
\newblock Control of deep-space formation-flying spacecraft; relative sensing
  and switched information.
\newblock \emph{Journal of Guidance, Control, and Dynamics}, 28(1), 106--114.

\bibitem[{Sundaram and Hadjicostis(2008)}]{SH:08}
Sundaram, S. and Hadjicostis, C.N. (2008).
\newblock Distributed function calculation and consensus using linear iterative
  strategies.
\newblock \emph{IEEE J.\ Sel.\ Areas Commun.}, 26(4), 650--660.

\bibitem[{Verma(1988)}]{V:88}
Verma, M.S. (1988).
\newblock Coprime fractional representations and stability of non-linear
  feedback systems.
\newblock \emph{Int.\ J.\ Control}, 48, 897--918.

\bibitem[{Vidyasagar(1985)}]{V:85}
Vidyasagar, M. (1985).
\newblock \emph{Control System Synthesis: A Factorization Approach}.
\newblock The MIT Press, Cambridge, MA.

\bibitem[{Wieland et~al.(2011)Wieland, Sepulchre, and Allg{\"o}wer}]{WSA:11}
Wieland, P., Sepulchre, R., and Allg{\"o}wer, F. (2011).
\newblock An internal model principle is necessary and sufficient for linear
  output synchronization.
\newblock \emph{Automatica}, 47(5), 1068--1074.

\bibitem[{Yucelen and Egerstedt(2012)}]{YE:12}
Yucelen, T. and Egerstedt, M. (2012).
\newblock Control of multiagent systems under persistent disturbances.
\newblock In \emph{"Proc.\ 2012 American Control Conf.,}, 5264--5269.

\bibitem[{Zelazo and Mesbahi(2011{\natexlab{a}})}]{ZM:11}
Zelazo, D. and Mesbahi, M. (2011{\natexlab{a}}).
\newblock Edge agreement: Graph-theoretic performance bounds and passivity
  analysis.
\newblock \emph{IEEE Transactions on Automatic Control}, 56(3), 544--555.

\bibitem[{Zelazo and Mesbahi(2011{\natexlab{b}})}]{ZM:11rsn}
Zelazo, D. and Mesbahi, M. (2011{\natexlab{b}}).
\newblock Graph-theoretic analysis and synthesis of relative sensing networks.
\newblock \emph{IEEE Transactions on Automatic Control}, 56(5), 971--982.

\bibitem[{Zhou et~al.(1996)Zhou, Doyle, and Glover}]{ZDG:95}
Zhou, K., Doyle, J.C., and Glover, K. (1996).
\newblock \emph{Robust and Optimal Control}.
\newblock Prentice-Hall, Englewood Cliffs, NJ.

\end{thebibliography}
\end{document}